\documentclass[11pt]{article}

\usepackage[usenames,dvipsnames,svgnames,table]{xcolor}
\usepackage{enumitem}
\usepackage{graphicx}
\usepackage{fancybox}
\usepackage{comment}
\usepackage{xcolor}
\usepackage{nameref}

\definecolor{ForestGreen}{rgb}{0.1333,0.5451,0.1333}
\definecolor{DarkRed}{rgb}{0.8,0,0}
\definecolor{Red}{rgb}{1,0,0}
\usepackage[linktocpage=true,
pagebackref=true,colorlinks,
linkcolor=ForestGreen,citecolor=ForestGreen,
bookmarks,bookmarksopen,bookmarksnumbered]
{hyperref}

\usepackage{amsmath}
\usepackage{amssymb}
\usepackage{amsthm}

\usepackage[ruled, noend, linesnumbered]{algorithm2e}

\usepackage{subfig}

\usepackage{thm-restate}

\usepackage{float}
\usepackage{cleveref}
\usepackage[bottom]{footmisc}

\newtheorem{theorem}{Theorem}[section]

\newtheorem{lemma}[theorem]{Lemma}

\newtheorem{remark}[theorem]{Remark}

\newtheorem*{theorem*}{Theorem}
\newtheorem*{corollary*}{Corollary}
\newtheorem*{conjecture*}{Conjecture}
\newtheorem*{lemma*}{Lemma}
\newtheorem*{thm*}{Theorem}
\newtheorem*{prop*}{Proposition}
\newtheorem*{obs*}{Observation}
\newtheorem*{definition*}{Definition}

\newtheorem*{remark*}{Remark}
\newtheorem*{rec*}{Recommendation}

\newenvironment{fminipage}%
  {\begin{Sbox}\begin{minipage}}%
  {\end{minipage}\end{Sbox}\fbox{\TheSbox}}

\newcommand\dd{\boldsymbol{\mathit{d}}}

\newcommand\ff{\boldsymbol{\mathit{f}}}
\renewcommand\gg{\boldsymbol{\mathit{g}}}

\newcommand\vv{\boldsymbol{\mathit{v}}}

\newcommand\vecone{\boldsymbol{1}}

\renewcommand\AA{\boldsymbol{\mathit{A}}}
\newcommand\BB{\boldsymbol{\mathit{B}}}
\newcommand\CC{\boldsymbol{\mathit{C}}}

\newcommand{\E}[1]{\mathop{{}\mathbb{E}}\left[#1\right]}

\DeclareMathOperator*{\argmin}{arg\,min}

\newcommand{\Center}{\textsc{Center}}
\newcommand{\diam}{diam}
\renewcommand{\P}{\mathbb{P}}
\renewcommand{\E}{\mathbb{E}}

\DeclareMathOperator{\dist}{dist}

\usepackage{fullpage}

\newcommand*\samethanks[1][\value{footnote}]{\footnotemark[#1]}

\title{Random-Shift Revisited: Tight Approximations for Tree Embeddings and $\ell_1$-Oblivious Routings}
\author{
 Rasmus Kyng
 \thanks{The research leading to these results has received funding from grant no. 200021 204787 of the Swiss National Science Foundation.} \thanks{The research leading to these results has received funding from the starting grant ``A New Paradigm for Flow and Cut Algorithms'' (no. TMSGI2 218022) of the Swiss National Science Foundation.}
 \\
 ETH Zurich \\
kyng@inf.ethz.ch
\and
Maximilian Probst Gutenberg\samethanks[1]\\
 ETH Zurich \\
maximilian.probst@inf.ethz.ch
\and
Tim Rieder\\
 ETH Zurich \\
tim.rieder@inf.ethz.ch}
\date{}

\begin{document}
\maketitle
\begin{abstract}

We present a new and surprisingly simple analysis of random-shift decompositions—originally proposed by Miller, Peng, and Xu [SPAA’13]: We show that decompositions for exponentially growing scales \(D = 2^0, 2^1, \ldots, 2^{\log_2(\operatorname{diam}(G))}\), have a tight \emph{constant}
trade-off between distance-to-center and separation probability \emph{on average} across the distance scales -- opposed to a necessary $\Omega(\log n)$ trade-off for a single scale.

This almost immediately yields a way to compute a tree $T$ for graph $G$ that preserves all graph distances with expected \(O(\log n)\)-stretch. This gives an alternative proof that obtains tight approximation bounds of the seminal result by Fakcharoenphol, Rao, and Talwar [STOC’03] matching the $\Omega(\log n)$ lower bound by Bartal [FOCS'96]. Our insights can also be used to refine the analysis of a simple $\ell_1$-oblivious routing proposed in [FOCS'22], yielding a tight $O(\log n)$ competitive ratio. 

Our algorithms for constructing tree embeddings and $\ell_1$-oblivious routings can be implemented in the sequential, parallel, and distributed settings with optimal work, depth, and rounds, up to polylogarithmic factors. Previously, fast algorithms with tight guarantees were not known for tree embeddings in parallel and distributed settings, and for $\ell_1$-oblivious routings, not even a fast sequential algorithm was known.
\end{abstract}

\pagenumbering{gobble}

\pagebreak

\pagenumbering{arabic}

\section{Introduction}

Distances in graphs lie at the heart of numerous applications in computer science and related fields. Whether modeling communication in social networks~\cite{distance-queries-social-networks-delling2014,distance-queries-social-networks-basu2024}, finding routings of almost optimal congestion in networks~\cite{distance-queries-congestion-minimization-racke2008}, or analyzing biological networks~\cite{distance-queries-biological-networks-pavlopoulos2011}, efficient distance queries are indispensable. More recently, there has been a rapid increase in interest in optimal transport for machine learning purposes (domain adaptation and generative models, text embeddings)~\cite{optimal-transport-machine-learning-lee2018, optimal-transport-generative-machine-learning-arjovsky2017}, computer graphics~\cite{optimal-transport-graphics-lavenant2018}, and analyzing biological data~\cite{optimal-transport-biology-gene-expression-schiebinger2019}, which can be solved in linear time on trees~\cite{optimal-transport-fast-trees-chen2024}.

In many practical scenarios, approximate distances are often sufficient, and rapid query responses are prioritized over exact computations. A particularly elegant solution to increasing speed and drastically reducing storage is the use of \emph{probabilistic tree embeddings}. Given input graph $G = (V,E,w : E \mapsto [1, W])$ with $n := |V|$ and $m := |E|$ and $W$ polynomially-bounded in $n$. A tree $T$ drawn from some distribution $\mathcal{T}$ is an $\alpha$-approximate probabilistic tree embedding of $G$ if $T$ preserves all distances in expectation by at most an $\alpha$-factor, i.e. for any $u,v \in V$, $\mathbb{E}_{T \sim \mathcal{T}}[\dist_T(u,v)] \leq \alpha \cdot\dist_G(u,v)$\footnote{We note that a deterministic notion of probabilistic tree embeddings has been considered where the goal is to find a tree that minimizes the average stretch among all edges}
and $\dist_T(u,v) \geq \dist_G(u,v)$. Probabilistic tree embeddings yield a particularly simple representation of the distances of the underlying graph or metric. Further, many (NP-)hard problems on graphs connected to cuts, network flow, metric labeling, buy-at-bulk network design, and vertex cover become tractable on trees~\cite{metric-labeling-chekuri2001, metric-labeling-kleinberg2002, buy-at-bulk-awerbuch1997}, thus yielding $O(\alpha)$-approximate algorithms to many important graph problems.

\paragraph{Probabilistic tree embeddings.} The notion of probabilistic tree embeddings was first studied in \cite{buy-at-bulk-awerbuch1997}. The problem of finding good tree embeddings has since received considerable attention. In \cite{buy-at-bulk-awerbuch1997}, a first algorithm was given achieving an $O(e^{\sqrt{\log n \log\log n}})$ approximation. Since, a long line of work (see \cite{lower-bound-embeddings-rabinovich1998, previous-embeddings-alon1995, previous-embeddings-bartal1996, lognloglogn-embedding-bartal1998, frttrees2003, bartal2004graph} and the references therein) first yielded an algorithm that achieves polylogarithmic approximation \cite{previous-embeddings-bartal1996} and finally culminated in the seminal result by Fakcharoenphol, Rao, and Talwar~\cite{frttrees2003} giving a tree embedding with tight approximation $\Theta(\log n)$. In \cite{bartal2004graph}, an $O(\log n)$-approximate tree embedding is given via a proof that condenses the structural insights from \cite{frttrees2003}. Tree embeddings obtained via their algorithm or some variation have since been referred to as FRT trees. To date, FRT trees are widely popular as the algorithm is reasonably simple and the analysis is insightful. They have also been taught in various undergraduate courses on advanced algorithms.

\paragraph{$\ell_1$-oblivious routings.} An $\ell_1$-oblivious routing is a matrix $\AA \in \mathbb{R}^{E \times V}$ that given any demand $\dd \perp \vecone$, i.e. $\dd^\top \vecone = 0$, over the vertices yields a flow $\mathbf{f} = \AA \mathbf{d}$ that routes the demand. The competitive ratio of an oblivious routing $\AA$ is the worst ratio achieved over all demands $\mathbf{d}$ between the $\ell_1$-cost\footnote{This is also often referred to as the cost of a \emph{transshipment flow}.} of the flow $\mathbf{f}$ and the optimal $\ell_1$-cost $\textsc{OPT}(\mathbf{d})$ of any flow routing demand $\mathbf{d}$. Today, $\ell_1$-oblivious routings are also an essential building block in distance-based algorithms that employ convex optimization, as the occurring subproblem can be solved by few left- and right-multiplies with $\AA$ (see \cite{sherman2017generalized, li2020faster, zuzic2022universally, rozhovn2022deterministic, zuzic2023simple, fox2024simple}). 

Note that for any tree $T$, we can construct a matrix $\AA_T$ such that $\mathbf{d}$ is routed optimally on $T$. Thus, for fixed $\dd$, an $\alpha$-approximate probabilistic tree embedding $T$ yields that $\AA_T \dd$ has expected $\ell_1$-cost $\alpha$. It follows that for FRT trees $T$ being sampled from distribution $\mathcal{T}$ with probability $p_T$, we have that $\AA = \sum_{T \in \mathcal{T}} p_T \cdot \AA_T$ is $O(\log n)$ competitive for every demand, and thus, an $\ell_1$-oblivious routing of quality $O(\log n)$. This is tight by the lower bounds on probabilistic tree embeddings. Alternatively, a construction by Englert and Räcke \cite{englert2009oblivious} achieves a tight competitive ratio. 

\paragraph{Low-diameter decompositions and probabilistic tree embeddings.} Given a graph $G$ and a diameter $D$, a low-diameter decomposition is a partition $\mathcal{X}$ of the vertex set into clusters $X \in \mathcal{X}$ that separate any two vertices $u,v \in V$ with probability at most $O(\dist_G(u,v) / D)$ while every cluster has weak diameter $\diam_G(X) = \max_{x,y\in X} \dist_G(x,y)$ at most $O(D \log n)$. The trade-off between separation probability and diameter is known to be tight \cite{previous-embeddings-bartal1996}.

Probabilistic tree embeddings are intimately connected to low-diameter decompositions: a hierarchy of low-diameter decompositions $\mathcal{A}_0 = \{\{v\} \;|\; v \in V\}, \mathcal{A}_1,  \ldots, \mathcal{A}_L = \{V\}$ for $L = \lceil \log_2(nW) \rceil$ where each $\mathcal{A}_l$ is a low-diameter decomposition of $G$ for parameter $2^l$ can be used almost directly to induce a probabilistic tree. That is, letting $\mathcal{A}_{\geq l}$ be the coarsest common refinement of partitions $\mathcal{A}_l, \mathcal{A}_{l+1}, \ldots, \mathcal{A}_L$, we can construct a tree $T$ that has a node associated with each cluster $C$ in a partition $\mathcal{A}_{\geq l}$ and for $l < L$ and edge to the node associated with cluster $C \subseteq C' \in \mathcal{A}_{\geq l+1}$ of weight $diam_G(C')$. The approximation factor achieved is $O(\log^2 n)$ with one log-factor stemming from the (tight) trade-off between separation probability and the diameter bound on each level, and the second log-factor from the number of levels.  

In \cite{frttrees2003}, a hierarchical low-diameter decomposition is developed such that for every vertex $u \in V$, the trade-off between separation probabilities of $u$ and diameter of clusters containing $u$ costs a log-factor across \textbf{all} levels. Formally, letting $\mathcal{A}_l(u)$ denote the cluster of $u$ in $\mathcal{A}_l$, this property can be succinctly summarized as $\sum_{0 \leq l \leq L} \frac{\diam_G(\mathcal{A}_l(u))}{2^l} = O(\log n)$ (which implies  $\sum_{0 \leq l \leq L} \frac{\diam_G(\mathcal{A}_{\geq l}(u))}{2^l} = O(\log n)$).

\paragraph{Random-shift decompositions.} In this article, we consider a different type of low-diameter decompositions than \cite{frttrees2003}: random-shift decompositions -- first suggested by Miller, Peng and Xu \cite{randomshift2013}. Random-shift Decompositions are popular due to their simplicity and amenability to various computational settings. These decompositions are computed by the algorithm given in \Cref{alg:randomShift}. The algorithm first samples a negative delay for each vertex $u \in V$ independently from the exponential distribution with mean $D$. It then computes a clustering around vertices with large negative delay by running a single-source shortest path (SSSP) computation.

\begin{algorithm}[H]
\caption{\textsc{RandomShiftDecomposition}($G = (V, E, w : E \mapsto [1, W], D)$)}
\label{alg:randomShift}
\For{each vertex $u \in V$}{
     Pick $\delta^u$ independently from the exponential distribution with mean $D$.
    }
    Let $G^s$ be the graph obtained from adding dummy source $s$ to $G$ and an edge $(s,v)$ of weight $\max_{w \in V} \delta^w - \delta^v$.\\
    Call an exact SSSP procedure on $G^s$ from $s$ and let $T^s$ be the shortest path tree obtained rooted at $s$. \\
    \For{each vertex $u \in V$ that is a child of $s$ in $T^s$}{
         $C^u$ is the set of all vertices in the subtree of $T^s$ rooted at $u$.\\
         \lForEach{$v \in C^u$}{$\textsc{Center}(v) \gets u$.}
    }
    \Return $(\mathcal{C} = \bigcup_{u \in V} \{C^u \}\setminus \{\emptyset\}, \textsc{Center})$ 
\end{algorithm}

A simple but clever analysis (see \cite{randomshift2013}) yields that the separation probability for any $u,v$ can be bounded by $O(\dist_G(u,v)/D)$ as desired, and it is easy to glean from the exponential distribution that no negative delay exceeds $O(D \log n)$ w.h.p. which yields the desired diameter bound of $O(D \log n)$. 

Let $\mathcal{C}_0 = \{\{v\} \;|\; v \in V\}$ and each singleton cluster has the vertex as center, the sets $\mathcal{C}_l$ for $0 < l < L$ being obtained by calling  \Cref{alg:randomShift} with parameter $2^l$ and $\mathcal{C}_L = \{V\}$ where the center of $V$ is an arbitrarily chosen vertex in $V$. Then, we call $\mathcal{C}_0, \mathcal{C}_1, \ldots, \mathcal{C}_L$ a \emph{hierarchical random-shift decomposition}. Naturally, we can derive a tree embedding $T$ as described above from these decompositions with approximation $O(\log^2 n)$.
Later, \cite{blurryballgrowing-becker2019} showed how to compute such hierarchical random shift decompositions in parallel using approximate SSSP computations.

Because hierarchical random-shift decompositions are 
simple, efficient, and parallelization-friendly, it is natural to ask if they can yield even stronger results.
However, a priori, this might seem unlikely, as
it is easy to construct examples where for fixed $u \in V$, the diameter of every cluster $\mathcal{C}_l(u)$ for $0 < l < L$ is $\Omega(2^l \log n)$ and thus one can seemingly not hope for a tight tree embedding induced by these decompositions.

\paragraph{A slightly tighter tree embedding construction.} On further inspection, one can slightly tighten the tree construction used above by choosing centers: Let $\textsc{Center}_0, \textsc{Center}_1, \ldots,  \textsc{Center}_L$\footnote{Throughout, all center functions are consistent, meaning all vertices in the same cluster have the same center.} be the associated center functions for the hierarchical random-shift decompositions $\mathcal{C}_0, \mathcal{C}_1, \ldots, \mathcal{C}_L$ as above where $\textsc{Center}_0$ maps each vertex to itself, and $\textsc{Center}_L$ maps all vertices to an arbitrary fixed vertex $r \in V$. 

When forming the tree induced by the (refinements of) $\mathcal{C}_0, \mathcal{C}_1, \ldots, \mathcal{C}_L$, one can now set the weight of edges $(C, C')$ where $C \in \mathcal{C}_{\geq l}, C \subseteq C' \in \mathcal{C}_{\geq l+1}$ equal to $\dist_G(\textsc{Center}_l(v), \textsc{Center}_{l+1}(v))$ for some $v \in C$ instead of using the rather loose upper bound of $\diam_G(C')$. Further, by consistency of the center function and the triangle inequality, we then have for every $w \in C$, that 
\begin{align*}
\dist_G(\textsc{Center}_l(v), \textsc{Center}_{l+1}(v)) &= \dist_G(\textsc{Center}_l(w), \textsc{Center}_{l+1}(w)) \\&\leq \dist_G(\textsc{Center}_l(w), w) + \dist_G(w, \textsc{Center}_{l+1}(v)).
\end{align*}
This yields an $O(\log n)$-approximate tree embedding when the trade-off between separation probabilities and distances to cluster centers is constant on average, and thus $O(\log n)$ over all distance scales.

\subsection{Our Contribution}

In this article, we show that for hierarchical random-shift decompositions, the trade-off between separation probabilities and distances to centers is indeed bounded by $O(\log n)$.\footnote{It was since pointed out that a similar result was derived in \cite{czumaj2021exploiting}, Theorem 2. This result, in turn, tightens the analysis from \cite{haeupler2016faster} by an $O(\log \log n)$ factor.} We summarize this result in the theorem below.

\begin{theorem}\label{thm:randomShiftExactMain}
Given graph $G = (V,E,w : E \mapsto [1, W])$, let $\mathcal{C}_0, \mathcal{C}_1, \ldots, \mathcal{C}_L$ be a hierarchical random-shift decomposition as described above. Then, for every $0 \leq l \leq L$ and $v \in V$, let  $\mathcal{C}_l(v)$ be the cluster containing $v$ in $\mathcal{C}_l$. Then, we have 
\begin{enumerate}
    \item for any $v \in V$, \[\mathbb{E}\left[\sum_{0 \leq l \leq L} \frac{\dist_G(v, \textsc{Center}_l(v))}{2^l}\right] \leq O(\log n). \]
    \item for any $0 \leq l \leq L$, $u,v \in V$, $\mathbb{P}[ \mathcal{C}_l(u) \neq \mathcal{C}_l(v)] = O\left(\frac{\dist_G(u,v)}{2^l}\right)$.
\end{enumerate}
\end{theorem}
\begin{remark}\label{rmk:nestingRSDecomp}
For $0 \leq l\leq L$, we define $\mathcal{C}_{\geq l}$ to be the coarsest common refinement of $\mathcal{C}_l, \mathcal{C}_{l+1}, \ldots, \mathcal{C}_L$. For $v \in C \in \mathcal{C}_{\geq l}$, we define $\textsc{Center}_{\geq l}(v) = \argmin_{r \in C} \dist(r, \textsc{Center}_l(v))$. 

Then, the hierarchy $\mathcal{C}_0, \mathcal{C}_1, \ldots, \mathcal{C}_L$ along with functions $\textsc{Center}_{\geq 0}, 
 \textsc{Center}_{\geq 1}, \ldots, \textsc{Center}_{\geq L}$ satisfies the properties above and is refining. Therefore, we say it is a \emph{refining hierarchical random-shift decomposition}.
\end{remark}

Since the sets in $\mathcal{C}_{\geq 0}, \mathcal{C}_{\geq 1}, \ldots, \mathcal{C}_{\geq L}$ form a laminar family, they naturally induce a tree $T$, and it is not hard to show that $T$ forms an $O(\log n)$-approximate probabilistic tree embedding. 

We then show that a hierarchical \emph{approximate} random-shift decomposition $\widehat{\mathcal{C}}_0, \widehat{\mathcal{C}}_1, \ldots, \widehat{\mathcal{C}}_L$ as suggested in \cite{blurryballgrowing-becker2019} yields almost the same properties as stated in \Cref{rmk:nestingRSDecomp}: the only difference being that each clustering only forms a subpartition, and each vertex $v$ is only clustered with probability at least $1/2$ (and thus the first property only applies to levels where $v$ was clustered and thus a center for $v$ exists). We show further that, up to constant factors, the refinements $\widehat{\mathcal{C}}_{\geq 0}, \widehat{\mathcal{C}}_{\geq 1}, \ldots, \widehat{\mathcal{C}}_{\geq L}$ satisfy the same properties as its exact counterpart, given in \Cref{rmk:nestingRSDecomp}.

Since the algorithm from \cite{blurryballgrowing-becker2019} only requires approximate SSSP as a primitive, it can be implemented optimally up to logarithmic factors in the sequential, parallel, distributed, and semi-streaming settings\footnote{We refer the reader to \cite{blurryballgrowing-becker2019} for precise definitions of the PRAM, CONGEST and semi-streaming models. We further point the reader to \cite{rozhovn2022deterministic} for information about a simpler minor-aggregation model that can be implemented efficiently in PRAM and CONGEST models.}. This implies fast algorithms for probabilistic tree embeddings in these settings.

\begin{theorem}[Main Result for Probabilistic Tree Embeddings]\label{thm:unfiedAppraochTreeEmbeddings}
Given graph $G=(V,E,w)$, there is a randomized algorithm that returns a $O(\log n)$-approximate tree embedding $T$ of $G$. 

The algorithm can be implemented to run
\begin{enumerate}
    \item in time $O(m \log n)$ in the sequential model.
    \item in polylogarithmic depth and almost-linear work in the PRAM model.
    \item in $\tilde{O}(\sqrt{n} + \textsc{HopDiam}(G))$\footnote{In this article, we let $n$ denote the number of vertices of the input graph and write $\tilde{O}(f(k))$ to mean $O(f(k) \log^c(n))$ for any constant $c$.} rounds with message complexity $\tilde{O}(m)$ in the CONGEST model where $\textsc{HopDiam}(G)$ is the unweighted (i.e. hop) diameter of the graph $G$.
    \item in $\tilde{O}(1)$ rounds in the semi-streaming model with high probablility.
\end{enumerate}
\end{theorem}

The theorem above recovers the fastest runtime of any sequential algorithm to construct FRT trees \cite{blelloch2017efficient} while being significantly simpler; the first parallel algorithm that simultaneously achieves work $\tilde{O}(m)$ and $\tilde{O}(1)$ depth; and the first distributed algorithm matching the lower bound from \cite{das2011distributed} on the number of rounds up to polylogarithmic factors. In the latter two settings, this yields a polynomial improvement for non-dense graphs over the state of the art \cite{blelloch2020parallelism, friedrichs2018parallel, ghaffari2014near}.

Further, we show that the $\ell_1$-oblivious routings from \cite{rozhovn2022deterministic} have $O(\log n)$-approximation. While we essentially use their construction, we give a different proof that leverages the insights from \Cref{thm:randomShiftExactMain} to obtain the tight approximation factor.

\begin{theorem}[Main Result for $\ell_1$-Oblivious Routing]
\label{thm:oblviousRouting}
Given graph $G=(V,E,w)$, there is a randomized algorithm that returns an $O(\log n)$-approximate $\ell_1$-oblivious routing $\AA$ of $G$ w.h.p. 

The algorithm can be implemented to run
\begin{enumerate}
    \item in time $\tilde{O}(m)$ in the sequential model.
    \item in polylogarithmic depth and almost-linear work in the PRAM model.
    \item in $\tilde{O}(\sqrt{n} + \textsc{HopDiam}(G))$ rounds with message complexity $\tilde{O}(m)$ in the CONGEST model.
\end{enumerate}
\end{theorem}

Notably, this yields the first $\tilde{O}(m)$ algorithm to construct $\ell_1$-oblivious routings of asymptotically optimal competitive ratio! While $\tilde{O}(m)$ time algorithms were previously known, all of them obtained rather large subpolynomial or polylogarithmic factors in the competitive ratio.

\paragraph{Roadmap.} We defer a detailed discussion of related work to \Cref{sec:sota}. In \Cref{sec:analysisExact}, we prove \Cref{thm:randomShiftExactMain}. In \Cref{sec:analysisApprox}, we prove that almost the same guarantees hold for approximate random-shift decompositions as suggested in \cite{blurryballgrowing-becker2019}. Finally, in \Cref{sec:treeEmbeddings}, we analyze the resulting tree embeddings, and in \Cref{sec:obliviousRouting}, we give a novel analysis of the $\ell_1$-oblivious routing construction suggested in \cite{rozhovn2022deterministic} yielding an $\ell_1$-oblivious routing with tight competitive ratio.

\section{Preliminaries}

\paragraph{Graphs.} We denote graphs as tuples $G = (V, E, w : E \mapsto [1,W])$ with $n:= |V|$ and $m := |E|$. In this article, all graphs are undirected and weighted. For a subset $A \subset V$, we use $G[A]$ to denote the vertex-induced graph and $E[A]$ to denote the set of edges with both endpoints in $A$, i.e. $E[A] = \{ \{u,v\} \in E, u,v \in A\}$. For disjoint subsets $A, B\subset V$, we denote by $E(A, B)$ edges in $E$ with exactly one endpoint in $A$ and $B$. We denote by $dist_G(s,t)$ the distance in $G$ under $w$ from $s$ to $t$. We let $B_G(v,r)$ denote the closed ball centered at $v$ with radius $r$. We often omit the subscript $G$ when the underlying graph $G$ is clear.

\paragraph{Exponential Distribution.} We recall that the exponential distribution with mean $D$, denoted by $\textsc{Exp}(D)$, has \emph{probability density function} (pdf) $p(x) = e^{- x/D} / D$ and \emph{cumulative density function} (cdf) $F(x) = 1 - e^{-x/D}$. 

The exponential distribution has the memoryless property, i.e. for $X \sim \textsc{Exp}(D)$, we have $\P[X > x + y | X > y] = \P[X > x]$ for all $x,y \geq 0$. Note further, that the exponential distribution satisfies that $X \sim \textsc{Exp}(1)$ iff $D \cdot X \sim \textsc{Exp}(D)$ for any $D > 0$.

\paragraph{Chernoff Bounds.} We use the classic Chernoff bound for i.i.d. random variables $X_1, X_2, \ldots, X_n \in [0,1]$, that for $\mu = \sum_i X_i$, we have for every $\delta > 0$, that $\P[|X - \mu| > \delta \mu] < 2\cdot e^{-\delta^2 \mu/3}$. 

We also use the Chernoff bound for i.i.d. random variables $X_1, X_2, \ldots, X_n \sim \textsc{Exp}(1)$ which yields that for $X = \sum_i X_i$, we have for any $\delta > 0$, that $\P[X > (2+\delta) n] < e^{-\delta n/2}$. Note that this bound does not make any assumptions on the values that variables $X_i$ can take. The bound can be derived straightforwardly via the moment-generating function of the exponential distribution. For a lack of an accessible resource, we give a short proof in \Cref{sec:chernoff}.

\paragraph{Logarithms.} We denote by $\log(x)$ the natural logarithm of $x$ and by $\lg(x)$ the logarithm of $x$ with base $2$.

\section{Random-Shift Decompositions via Exact SSSP}
\label{sec:analysisExact}

In this section, we give a proof of \Cref{thm:randomShiftExactMain} for the case when $\mathcal{C}_1, \mathcal{C}_2, \ldots, \mathcal{C}_{L-1}$ are computed via the standard random-shift algorithm from \Cref{alg:randomShift}. 

\paragraph{Correctness of the first property via a simple ball growing process.} At the heart of our new analysis lies a simple ball growing process over levels $l$. Fixing the vertex $v \in V$, we define the process by a sequence of integers $r_l$ such that balls $B(v, r_l 2^l)$ are growing monotonically.

Let $r_0 := 1$ with $|B(v,r_02^0)| \ge 1$ and 
let $r_l := r_{l-1}/2 + \Delta_l + 2$, where we choose $\Delta_l \;(l \ge 1)$ as the largest integer such that growing the radius of the ball can be paid for by an exponential growth in the number of vertices in it, namely
\begin{equation}\label{eq:setIncrease}
    \Delta_l := \max \xi \text{ satisfying } |B(v,(r_{l-1}/2+\xi+2) 2^l)| \ge e^{\xi/2} \cdot |B(v,r_{l-1}2^{l-1})|. 
\end{equation}
Note that $r_l2^l = (r_{l-1}/2 + \Delta_l + 2)2 \cdot 2^{l-1} > r_{l-1}2^{l-1}$, thus $r_l 2^l$ grows monotonically as claimed. 

In the remainder of this section, we prove the following two lemmas. The first lemma says that the distance to the center is well captured by $r_l 2^l$ on level $l$ (in expectation). The second yields that the sum of $r_l$'s is small, and thus, since $r_l$ can be thought of as the loss-per-level, we have a small loss overall. Combined, these lemmas yield Property 1 as a corollary. 
\begin{lemma}[Distances-to-Center] 
    \label{lem:distFromR}
    For any $0 \leq l \leq L$, $\mathbb{E}[\dist(v, \Center_l(v))] = O(r_l \cdot 2^l)$.
\end{lemma}

\begin{lemma}[Sum-of-Deviations] 
    \label{lem:RSLDD:sum-over-r-is-logarithmic}
    $\sum_{l=0}^{L} r_l \in O(\log n)$.
\end{lemma}

\paragraph{Bounding the distances to centers (Proof of \Cref{lem:distFromR}).} Let us fix any level $0 < l < L$ (for $l \in \{0, L\}$ the proof is trivial). We let the variables $\delta^u$ and $C^u$ refer to the values that these variables take in \Cref{alg:randomShift} when computing $\mathcal{C}_l$.

For every $i \geq 0$, we let $B_{i} = B(v,(r_{l} + i) 2^l)$ and define $I_{i} = B_{i+2} \setminus B_{i+1}$. We can now rewrite the expectation as
\begin{align*}
    \E[\dist(v, \textsc{Center}_l(v))] &= \sum_{w \in V} \P[v \in C^w] \cdot \dist(w,v)
    \leq O(r_l \cdot 2^l) + \sum_{i \geq 1} \sum_{w \in I_i} \P[v \in C^w] \cdot (i+1) \cdot 2^l.
\end{align*}
It remains to bound the probabilities that a vertex in $I_i$ becomes the center of $v$. Note that a vertex $w \in I_i$ cannot become a center if for some $u \in B_0$, $\delta^u + i \cdot 2^l > \delta^w$, since then $v$ prefers $u$ over $w$ as a center. We obtain for every $w \in I_i$
{
\begin{align}\label{eq:mainAnalysisIntegration}\begin{split}
    \P[v \in C^w] &
    \leq \P[\max_{u \in B_0} \delta^u < \delta^w - i \cdot 2^l]\\
    &= \int_{x = i \cdot 2^l}^{\infty} \frac{1}{2^l} e^{-x/2^l} \cdot \P[\max_{u \in B_0} \delta^u + i \cdot 2^l < x] dx\\
    &= \int_{x = i \cdot 2^l}^{\infty} \frac{1}{2^l} e^{-x/2^l} \cdot (1-e^{-x/2^l + 
 i})^{|B_0|}dx\\
    &= \int_{x = 0}^{\infty} \frac{1}{2^l} e^{-x/2^l - i} \cdot (1-e^{-x/2^l})^{|B_0|}dx\\
    &= e^{-i} \int_{x = 0}^{\infty} \frac{1}{2^l} e^{-x/2^l} \cdot (1-e^{-x/2^l})^{|B_0|}dx\\
     \end{split}\end{align}
 \begin{align*}\begin{split}
    &\leq e^{-i} \int_{x = 0}^{\infty} \frac{1}{2^l} e^{-x/2^l} \cdot e^{- |B_0| \cdot e^{-x/2^l}}dx\\
    &= e^{-i} \cdot \left[ \frac{e^{- |B_0| \cdot e^{-x/2^l}}}{|B_0|} \right]_{x = 0}^{\infty}\\
    &=  e^{-i} \cdot \left(\frac{1}{|B_0|} - \frac{e^{- |B_0|}}{|B_0|}\right) \\ 
    &< \frac{e^{-i}}{|B_0|}.
\end{split}\end{align*} where in the first equality, we compute the probability by integrating over all values $x$ that $\delta^w$ can take. At each such value $x$, we take the density function $\frac{1}{2^l} e^{-x/2^l}$ of the exponential distribution times the probability that the maximum value $\delta^u$ over all $u \in B_0$ is smaller than $x$. Note that we start integrating only over $x \geq i \cdot 2^l$, since for smaller values, $\delta^w$ cannot exceed $\max_u \delta^u + i \cdot 2^l$ since $\delta^u \geq 0$.

In the second inequality, we use the probability that the maximum value over all variables $\delta^u$ is at most $x$ is given by the cumulative density function at value $x$ to the power $|B_0|$.

We substitute the value of $x$ in the third equality to simplify the expression, and pull out the factor $e^{-i}$ from the integral in the forth. We then use $1+y \leq e^y$ in the second inequality, and finally evaluate the integral.

We can therefore now conclude the proof using the above derivation combined with the bound $|I_i| < e^{(i+2)/2} |B_0|$ that follows from the growth process and $I_i \subseteq B_{i+2}$, which yields
\begin{align}\label{eq:concludeDistToCenter}
\begin{split}
     \E[\dist(v, \textsc{Center}_l(v)] &= O(r_l \cdot 2^l) + \sum_{i \geq 1} \sum_{w \in I_i} \P[v \in C^w] \cdot (i+1) \cdot 2^l\\
    &\leq O(r_l \cdot 2^l) + \sum_{i \geq 1} \sum_{w \in I_i}  \frac{e^{-i}}{|B_0|} \cdot (i+1) \cdot 2^{l}\\
    &= O(r_l \cdot 2^l) + \sum_{i \geq 1} |I_i| \cdot \frac{e^{-i}}{|B_0|} \cdot (i+1) \cdot 2^{l}\\
    &\leq O(r_l \cdot 2^l) + \sum_{i \geq 1} e^{(i+2)/2}|B_0| \cdot \frac{e^{-i}}{|B_0|} \cdot (i+1) \cdot 2^{l}\\
    &\leq O(r_l \cdot 2^l) + e \cdot \sum_{i \geq 1} e^{-i/2} \cdot (i+1) \cdot 2^{l}\\
     &= O(r_l \cdot 2^l)
\end{split}
\end{align}
by straightforward calculations and using that the last sum is geometric.

We point out that inspection of \eqref{eq:concludeDistToCenter} yields a slightly stronger statement: $\dist(v, \textsc{Center}_l(v))/2^l$ is stochastically dominated by $O(r_l + X_l)$ for some random variable $X_l \sim \textsc{Exp}(1)$.

\paragraph{Bounding the sum of deviations (Proof of \Cref{lem:RSLDD:sum-over-r-is-logarithmic}).} 
From the definition of $\Delta_l$ in \eqref{eq:setIncrease}, we have $|B(v,r_LD_L)| \ge e^{\Delta_L/2} \cdot |B(v,r_{L-1}D_{L-1}| \ge e^{\Delta_L/2} \cdot e^{\Delta_{L-1}/2} \cdots |B(v,r_0D_0)|$. But since there are at most $n$ vertices, we have $n \geq  |B(v,r_LD_L)| \ge |B(v,r_0D_0)| \cdot \prod_{l=1}^{L} e^{\Delta_l/2} \ge e^{\sum_{l=1}^{L} \Delta_l/2}$. Thus, $\sum_{l=1}^{L} \Delta_l \le 2 \ln n$. We thus get
\begin{align}
\nonumber
    \sum_{l=0}^{L} r_l &=\sum_{l=0}^{L} \left( r_0/2^l + \sum_{i=1}^l (\Delta_i+2)/2^{l-i} \right) &&\text{By definition of $r_l$}\\
\nonumber
    &= \sum_{l=0}^{L} r_0/2^l + \sum_{i=1}^{L} \sum_{l=i}^{L} (\Delta_i+2)/2^{l-i}  &&\text{Reordering the summation}\\
\nonumber
    &\le \sum_{l=0}^\infty r_0/2^l + \sum_{i=1}^{L} \sum_{l'=0}^\infty (\Delta_i+2)/2^{l'} &&\text{Using }l'=l-i\\
\label{eq:sumRbound}
    &= O(\log n).
\end{align} 

\paragraph{Correctness of the second property.} The second property, that $\mathbb{P}[ \mathcal{C}_l(u) \neq \mathcal{C}_l(v)] = O\left(\frac{\dist(u,v)}{2^l}\right)$ for all $0 \leq l \leq L$ and $u,v\in V$, follows from the definitions of $\mathcal{C}_0$ and $\mathcal{C}_L$ and the following insight first exploited in \cite{randomshift2013}.


\begin{lemma}[Lemma 4.4 of \cite{randomshift2013}]
\label{lem:gap1stvs2nddist}
Let \(d_{1} \le d_{2} \le \dots \le d_{n}\) be arbitrary values and let \(\delta_{1}, \dots, \delta_{n}\)
be independent random variables picked from \(\mathrm{Exp}(D)\),
and define $X_i = d_i - \delta_i$,
and define $X^{(1)} \leq X^{(2)} \leq \ldots \leq X^{(n)}$ to be the sequence of $\{X_i\}$ but sorted by increasing value. Then the probability that $X^{(2)} - X^{(1)} \leq c$ is at most 
\(c/D\).
\end{lemma}
\begin{proof}[Sketch of proof]
We model each value $-X_i = \delta_i - d_i $ as the failure time of a light bulb following an exponential distribution with rate \(\beta\) which is turned on at time $-d_i$. The quantity \(X^{(1)} = \min_i X_i \) corresponds to the moment the last bulb burns out, and we want to estimate how close this is to the moment when the second‐to‐last bulb fails. Because exponential distributions are memoryless, once the second‐to‐last bulb fails, the remaining bulb’s remaining lifetime still follows \(\mathrm{Exp}(\beta)\). Therefore, the probability that the last bulb also fails within additional time \(c\) is \(1 - e^{-c\beta} \leq \beta c\). If the last bulb has not yet been turned on at the moment the second‐to‐last bulb fails, then the probability of a gap under \(c\) can only decrease. Consequently, the overall probability that the time difference between the last two bulb failures is less than \(c\) remains bounded by \(1 - e^{-c\beta}\).
\end{proof}

For vertex $v \in V$, let $d_i$ be the distance from $v$ to the $i$-th vertex in the graph $G$. Let $\delta^i$ be the random value associated with this vertex in \Cref{alg:randomShift}. 
Then, the above result implies via the triangle inequality that $B(v, \dist(u,v))$ is contained in cluster $\mathcal{C}(v)$ with probability at least $\dist(u,v)/(2D)$. 

\paragraph{Refining hierarchical random-shift decompositions (Proof of \Cref{rmk:nestingRSDecomp}).} We recall that $\mathcal{C}_{\geq l}$ is the coarsest common refinement of $\mathcal{C}_l, \mathcal{C}_{l+1}, \ldots, \mathcal{C}_L$ and for $v \in C \in \mathcal{C}_{\geq l}$, $\textsc{Center}_{\geq l}(v) = \argmin_{u \in C} \dist(r, \textsc{Center}_l(u))$.   For the first property, it suffices to use that $C \subseteq C' \in \mathcal{C}_l$ which implies that all vertices $w \in C$ have the same center $\textsc{Center}_l(v)$, and the triangle inequality to conclude that
\begin{align}\label{eq:triangleInequalityForCenterProj}
\begin{split}
\dist(v, \textsc{Center}_{\geq l}(v)) &\leq \min_{u \in C} \dist(\textsc{Center}_l(v), u)
+ \dist(\textsc{Center}_l(v), v) 
\\
&\leq 2 \cdot \dist(\textsc{Center}_l(v), v).
\end{split}
\end{align}

For the second property, note that vertices $u,v \in V$ are separated only in $\mathcal{C}_{\geq l}$ if they are separated in any of the clusters $\mathcal{C}_l, \mathcal{C}_{l+1}, \ldots, \mathcal{C}_L$ and thus with probability at most $\sum_{i = l}^L O(\dist(u,v) / 2^i) = O(\dist(u,v)/ 2^l)$.
\section{Random-Shift Decompositions via Approximate SSSP}
\label{sec:analysisApprox}

In this section, we show that a tweaked version of \Cref{alg:randomShift} that only requires calls to approximate SSSP yields the guarantees described in \Cref{thm:randomShiftExactMain}. Our algorithm is almost exactly the algorithm from \cite{blurryballgrowing-becker2019}, except that we save one level of recursion (this will suffice for us; for the application in \cite{blurryballgrowing-becker2019} this would not yield the desired guarantees). We choose $\epsilon = \frac{1}{40 \log n}$ throughout this section.

\begin{algorithm}[H]
\caption{\textsc{ApproxRandomShiftDecomposition}($G = (V, E, w : E \mapsto [1, W], D)$)}
\label{alg:approxRandomShift}
\For{each vertex $u \in V$}{
     Pick $\delta^u$ independently from the exponential distribution with mean $D$.
    }
    Let $G^s$ be the graph obtained from adding dummy source $s$ to $G$ and an edge $(s,v)$ of weight $\max_{w \in V} \delta^w - \delta^v$.\\
    Call a {\color{blue} $(1+\epsilon)$-approximate SSSP} procedure on $G^s$ from $s$ and let $T^s$ be the {\color{blue} approximate} shortest path tree obtained rooted at $s$.
    \label{lne:apxSSSPtree}
    \\
    \For{each vertex $u \in V$ that is a child of $s$ in $T^s$}{
         $\widetilde{C}^u$ is the set of all vertices in the subtree of $T^s$ rooted at $u$.\label{lne:approxRSCluster}\\
         \lForEach(\label{lne:partitionapproxRSforeachAssignCenter}){$v \in \widetilde{C}^u$}{$\widetilde{\textsc{Center}}(v) \gets u$.}
         
         {\color{blue}$\widehat{C}^u \gets \widetilde{C}^u \setminus \textsc{Blur}(G, V \setminus \widetilde{C}^u, D)$. \label{lne:callBlur}}\\         \lForEach(\label{lne:approxRSforeachAssignCenter}){$v \in \widehat{C}^u$}{$\widehat{\textsc{Center}}(v) \gets u$.}
        }
    \Return $(\widehat{\mathcal{C}} = \bigcup_{u \in V} \{\widehat{C}^u \}\setminus \{\emptyset\}, \widehat{\textsc{Center}})$ 
\end{algorithm}

\begin{algorithm}
\caption{\textsc{Blur}$(G, X, D)$
}
\label{alg:blur}
$\widehat{X} \gets X$.\\
\For{$i = 0,1, \ldots, \lceil \log_{1/\epsilon}(D) \rceil$}{
    Let $G^i$ be the graph obtained from $G$ by contracting $\widehat{X}$ into a super-vertex.\\
    Let $r^i$ be sampled uniformly from $[0, \epsilon^{i}D/64]$.\\
    Call a {$(1+\epsilon^2)$-approximate SSSP} procedure on $G^i$ from $\widehat{X}$ and let $T^i$ be the {approximate} shortest path tree obtained rooted at $\widehat{X}$.\\
    Add to $\widehat{X}$ all new vertices in $B_{T^i}(\widehat{X}, r^i)$
}
\Return $\widehat{X}$.
\end{algorithm}

\paragraph{Preliminaries for the analysis.} We commonly refer to $\mathcal{C}$ and $\textsc{Center}$ as the objects obtained from running  \Cref{alg:randomShift} for the same parameter and random choices.

We further observe that when all values $\delta^u \leq D \cdot 9\log n$, then all distances in $G_s$ are at most $D \cdot 10 \log n$ and thus by choice of $\epsilon$ at most an additive error of $D/4$ in any computed distance (and also $T_s$ introduces at most additive slack $D/4$ in the distances). Since all of these events hold w.p. $1 - n^{-8}$, we condition on it for the rest of this section.

\paragraph{Each vertex $v$ is clustered with constant probability.} We first prove the following statement that trivially implies that every vertex is clustered with at least constant probability.

\begin{lemma}\label{lma:ballInCluster}
For every vertex $v \in V$, \Cref{alg:approxRandomShift} invoked with parameter $D$ returns $\widehat{\mathcal{C}}$ such that $B(v, D/8)$ is contained in some cluster in $\widehat{\mathcal{C}}$ with probability at least $1/2$.
\end{lemma}
\begin{proof}
Recall that $\mathcal{C}$ is the clustering obtained from running \Cref{alg:randomShift} with the same parameter and random choices. 
\Cref{lem:gap1stvs2nddist} tells us that with 
probability $\geq 1/2$, the path $s \to \textsc{Center}(v) \leadsto v $ is $D/2$ shorter than any other path $s \to w \leadsto v$, and thus $B(v, D/4) \in \mathcal{C}(v)$ as well.
But note that since there is at most an additive error $D/16$ in the distances, this event implies that $B(v, D/8 + D/16) \subseteq \widetilde{C}(v)$. Finally, we note that \Cref{alg:blur} only removes vertices $v \in \widetilde{\mathcal{C}}(v)$ with $\dist(v, V \setminus \widetilde{\mathcal{C}}(v)) \leq (1+\epsilon) \sum_{i=0}^{\lceil\log_{1/\epsilon}(D)\rceil} r^i \leq 2 \sum_{i=0}^{\infty} \epsilon^i D/64 < D/16$. Thus, with probability at least $1/2$, $v$ is not removed in this process, as desired.
\end{proof}

\paragraph{Correctness of the first property.}
We follow closely the proof from
\Cref{sec:analysisExact}: we use the same ball growing process, i.e. the ball w.r.t. the real graph distances. For fixed vertex $v \in V$, this yields the sequence $\{r_l\}_{0 \leq l \leq L}$. \Cref{lem:RSLDD:sum-over-r-is-logarithmic} remains true since it only used values $r_l$.

We then show that for every level $0 < l < L$, if $v$ is clustered by $\mathcal{C}_l$, then $\frac{\dist(v, \widehat{\textsc{Center}}_{l}(v))}{2^l}$ is stochastically dominated by $O(r_l + X_l)$ for some random variable $X_l \sim \textsc{Exp}(1)$ (that only depends on $\mathcal{C}_l$). Let us now fix any such level $0 < l < L$. 
For every $i \geq 0$, we let $B_{i} = B(v,(r_{l} + i) 2^l)$ and define $I_{i} = B_{i+2} \setminus B_{i+1}$. We can now again rewrite the expectation as
\begin{align*}
    \E[\dist(v, \widetilde{\textsc{Center}}_l(v)] &= \sum_{w \in V} \P[v \in \widetilde{C}^w] \cdot \dist(w,v)
    \leq O(r_l \cdot 2^l) + \sum_{i \geq 1} \sum_{w \in I_i} \P[v \in \tilde{C}^w] \cdot (i+1) \cdot 2^l.
\end{align*}
Finally, we have to slightly adapt our calculation of an upper bound on $\P[v \in \tilde{C}^w]$. Namely, since there is an additive error of up to $2^l$ in the distance computation, $w \in I_i$ cannot become a center of $v$ if for some $u \in B_0, \delta^u + i \cdot 2^l > \delta^w + 2^l$, i.e. we weaken the inequality in the event by an additive $2^l$ compared to the proof in \Cref{sec:analysisExact}. But this still yields that 
\begin{align}\label{eq:approxMainAnalysisIntegration}\begin{split}
    \P[v \in \tilde{C}^w] &
    \leq \P[\max_{u \in B_0} \delta^u < \delta^w - (i-1) \cdot 2^l]\\
    &= \int_{x = (i-1) \cdot 2^l}^{\infty} \frac{1}{2^l} e^{-x/2^l} \cdot \P[\max_{u \in B_0} \delta^u + (i-1) \cdot 2^l < x] dx\\
    &< \frac{e^{-i+1}}{|B_0|}.
\end{split}\end{align}
This matches the bound of \Cref{eq:mainAnalysisIntegration} up to a factor $e$,
allowing us to recover the bound of \Cref{eq:concludeDistToCenter}. The rest of the proof is analogous, yielding that $\E[\dist(v, \widetilde{\textsc{Center}}_l(v)]/2^l$ is stochastically dominated by $O(r_l + X_l)$ for some $X_l \sim \textsc{Exp}(1)$. 

Finally, $\E[\dist(v, \widehat{\textsc{Center}}_l(v)]/2^l$ is stochastically dominated by $O(r_l + X_l)$, since 
$\widehat{C}^u \subseteq \widetilde{C}^u$ and we do not define the distance for vertices with no assigned center.


\paragraph{The first property for refining partitions.}
Let us denote the event that $v$ is not in a cluster in $\mathcal{C}_j$ by $\mathcal{E}_j$. It remains to observe that $\widehat{\textsc{Center}}_{\geq l}(v)$ is found by projecting the center $\widehat{\textsc{Center}}_{l'}(v)$ for the smallest level $l'\geq l$ where the event $\mathcal{E}_{l'}$ holds. As in \eqref{eq:triangleInequalityForCenterProj} and our previous reasoning, we have by the triangle inequality, that $\dist(v, \widehat{\textsc{Center}}_{\geq l}(v) \leq 2 \cdot \dist(v, \widehat{\textsc{Center}}_{l'}(v))$. We conclude
\begin{align*}
\E&\left[\sum_{0 < l < L} \frac{\dist(v, \widehat{\textsc{Center}}_{\geq l}(v))}{2^l}\right] \\
&\leq \sum_{0 < l < L} \sum_{0 \leq i \leq L - l} \P[\cap_{l \leq j < l+i} \mathcal{E}_j \land \neg \mathcal{E}_{l+i}] \cdot 2 \cdot \E[\dist(v, \widehat{\textsc{Center}}_{l+i}(v)) / 2^{l+i} \mid \neg \mathcal{E}_{l+i}] \\
&\leq \sum_{0 < l < L} \sum_{0 \leq i \leq L - l} \prod_{l \leq j < l+i}\P[ \mathcal{E}_j ] \cdot 2 \cdot \E[\dist(v, \widehat{\textsc{Center}}_{l+i}(v)) / 2^{l+i} \mid \neg \mathcal{E}_{l+i}] \\
    &\leq \sum_{0 < l < L} \sum_{0 \leq i \leq L - l} 2^{-i} \cdot 2 \cdot \E[\dist(v, \widehat{\textsc{Center}}_{l+i}(v)) / 2^{l+i} \mid \neg \mathcal{E}_{l+i}] \\
&< \sum_{0 < l < L} \sum_{0 \leq i \leq \infty} 2^{-i} \cdot 2 \cdot \E[\dist(v, \widehat{\textsc{Center}}_{l}(v))/ 2^l \mid \neg \mathcal{E}_{l}] \\
&= O(1) \cdot \E\left[\sum_{0 \leq l \leq L, {\widehat{\mathcal{C}}_l(v) \neq \emptyset}} \dist(v, \widehat{\textsc{Center}}_{l}(v))/2^l \right]\\
&=O(\log n).    
\end{align*}
Where we use that the events $\mathcal{E}_j$ are independent as levels are independent from each other, and so are events $\mathcal{E}_j$ from distances-to-centers for level $\neq j$. We then use that $\mathcal{E}_j$ occurs with probability at most $1/2$ by \Cref{lma:ballInCluster}, observe a geometric sum, and appeal to the bound on the distances to centers in $\widehat{C}_l$ from the last paragraph.

\paragraph{Correctness of the second property.} We fix an arbitrary level $0 < l < L$, and vertices 
$v,w \in V$,
and we want to show 
$\mathbb{P}[ \widehat{\mathcal{C}}_l(w) \neq \widehat{\mathcal{C}}_l(v)] = O\left(\frac{\dist_G(w,v)}{2^l}\right)$. We prove this property by a simple case distinction. We use $\Delta = \dist_G(v,w)$.
\begin{itemize}
    \item \underline{if $\widetilde{C}(w) \neq \widetilde{C}(v)$:} Consider the call to  \Cref{alg:blur} in the for-loop iteration for cluster $\widetilde{C}^u$ with $w \in \widetilde{C}^u$. The sub-procedure initializes $\hat{X} \gets V \setminus \widetilde{C}(w)$. Since by assumption $v \not\in \widetilde{C}^u$, initially $\dist_G(\hat{X}, u) \leq \dist_G(u,w)$. It chooses $r^0 \geq 2\Delta$ with probability $1 - \frac{128 \cdot \Delta}{2^l}$. Since it adds $B(\hat{X}, r^0/2) \subseteq B_{T^0}(\hat{X}, r^0)$ to $\hat{X}$ (we use $\frac{1}{1+\epsilon} \leq 2$), the above event results in $v,w$ being both added to $\hat{X}$. 
    
    Since $\hat{X}$ is only increasing in later steps, and $\widehat{C}^w = \widetilde{C}^w \setminus \hat{X}$ for the final set $\hat{X}$, we have that $w$ is not clustered with probability $1 - \Omega(\Delta/2^l)$. The same argument applies to $v$ by symmetry, and thus, the probability that $v,w$ are both not clustered is $1-\Omega(\Delta/2^l)$. Thus, the $u$ and $v$ are only clustered and separated with probability $O(\Delta/2^l)$. 

    \item \underline{otherwise (if $\widetilde{C}(w) = \widetilde{C}(v)$):} in this case, the property follows immediately from the following lemma due to \cite{blurryballgrowing-becker2019}. We point out that the precise statement in \cite{blurryballgrowing-becker2019} only bounds the probability of endpoints of edges being separated, however, the statement below is implicit.

    \begin{lemma}[see Lemma 3.1 in \cite{blurryballgrowing-becker2019}]
    For $n \geq 2$, \Cref{alg:blur} with parameters $G, X, D$ returns a set $\hat{X} \supset X$ such that any vertices $u,v \in V$ are separated by $\hat{X}$ (i.e. one is contained, one is not), with probability $O(\dist(u,v)/D)$.  
    \end{lemma}
\end{itemize}

\section{Tight probabilistic tree embeddings}
\label{sec:treeEmbeddings}

\paragraph{Tree construction.} Let $\mathcal{R}_{\geq 0} = \{ \{v \} \mid v \in V\}, \mathcal{R}_1, \ldots, \mathcal{R}_{\geq L} = \{V\}$ be a \emph{refining hierarchical (approximate) random-shift decomposition} with the properties described in \Cref{rmk:nestingRSDecomp} as obtained in \Cref{sec:analysisExact} and \Cref{sec:analysisApprox}. For convenience, we extend the associated functions $\textsc{Center}_{\geq l}$ to yield for each cluster $C \in \mathcal{R}_{\geq l}$ the center $\textsc{Center}_{\geq l}(C)$ (since all vertices in $C$ agree on the same center, the image consists only of a single vertex). Now, consider the tree $T$ that has a node associated with each cluster $C \in \mathcal{R}_{\geq l}$, and for $l < L$, an edge from the node to the node associated with $C \subseteq C' \in \mathcal{R}_{\geq l+1}$ of weight $dist_G(\textsc{Center}_{\geq l}(C), \textsc{Center}_{\geq l+1}(C'))$. 

\paragraph{Correctness.} We next show that $T$ is an $O(\log n)$-approximate probabilistic tree embedding: fix any pair of vertices $u,v \in V$. Let $0 < l \leq L$ be the smallest integer, such that $u$ and $v$ end up in the same cluster $C$ in $\mathcal{R}_{\geq l}(v)$, i.e. $\mathcal{R}_{\geq l}(u) = \mathcal{R}_{\geq l}(v)$. Note that because the partitions are refining, $u$ and $v$ must also be in the same cluster at every higher level. 

Then, it is not hard to see that the distance from $u$ to $v$ consists of the weight of the paths from $u$ to the center $\textsc{Center}_{\geq l}(u) = \textsc{Center}_{\geq l}(v)$ and then to $v$ in $T$. It is thus not hard to see that we can write the expected distance from $u$ to $v$ in $T$ as
\begin{align*}
\E[\dist_T(u,v)] &= \sum_{0 \leq l < L} \P[\mathcal{R}_{\geq l}(u) \neq \mathcal{R}_{\geq l}(v)] \cdot \left(\sum_{z \in \{u,v\}}\dist_G(\textsc{Center}_{\geq l}(z), \textsc{Center}_{\geq l+1}(z)) \right)\\
&\leq 2 \cdot \sum_{0 \leq l < L} \P[\mathcal{R}_{\geq l}(u) \neq \mathcal{R}_{\geq l}(v)] \cdot \left(\sum_{z \in \{u,v\}}\dist_G(z, \textsc{Center}_{\geq l+1}(z)) \right)\\
&\leq 2 \cdot \sum_{0 \leq l < L} O(\dist_G(u,v)/2^l) \cdot \left(\sum_{z \in \{u,v\}}\dist_G(z, \textsc{Center}_{\geq l+1}(z)) \right)\\
&\leq O(\dist_G(u,v)) \cdot \sum_{0 \leq l < L} \left(\sum_{z \in \{u,v\}} \frac{\dist_G(z, \textsc{Center}_{\geq l+1}(z))}{2^l} \right)\\
&= O(\dist_G(u,v)) \cdot O(\log n).
\end{align*}
where we use in the first inequality that by the triangle inequality 
\[
\dist_G(\textsc{Center}_{\geq i}(z), \textsc{Center}_{\geq i+1}(z)) \leq \dist_G(\textsc{Center}_{\geq i}(z), z) + \dist_G(z, \textsc{Center}_{\geq i+1}(z))
\]
and that $\dist_G(\textsc{Center}_{\geq 0}(z), z) = 0$ because $\textsc{Center}_{\geq 0}(z) =z$; in the second inequality, we use the second property of  \Cref{rmk:nestingRSDecomp} (or rather \Cref{thm:randomShiftExactMain}); in the third, we re-arrange terms, and in the final inequality, we exchange sums and use the first property of  \Cref{rmk:nestingRSDecomp} for each $z \in \{u,v\}$ separately.

Note that in the construction of $T$, we used access to the exact distance between cluster centers. While these distances are readily available from the computations in \Cref{alg:randomShift}, when instead using the approximate \Cref{alg:approxRandomShift}, we only have access to the \emph{approximate} distances. However, approximate distances can only increase the stretch by a constant factor, so the asymptotic guarantee remains unchanged.

\paragraph{Implementation.} Our analysis above, and inspecting \Cref{alg:randomShift} straightforwardly allows us to conclude that in the sequential model of computation, only $O(\log n)$ SSSP computations on an $O(m)$ edge graph suffice and that the additional time to construct $T$ is at most $O(m \log n)$. Using the linear-time SSSP algorithm by Thorup \cite{sssp-thorup1999}, we thus obtain runtime $O(m \log n)$. This yields the claim for the sequential setting in \Cref{thm:unfiedAppraochTreeEmbeddings}.

In the parallel, distributed, and semi-streaming models, we refer the reader to \cite{blurryballgrowing-becker2019} for details on how to implement \Cref{alg:approxRandomShift} and tree $T$ efficiently. We merely point out, that since the publication of \cite{blurryballgrowing-becker2019}, the state-of-the-art to compute approximate SSSP in parallel has been significantly advanced (see \cite{li2020faster, asssp-improved-rozhoň2022}) which speeds-up their parallel implementation by exchanging the result from \cite{asssp-improved-rozhoň2022} with their parallel approximate SSSP implementation. \cite{asssp-improved-rozhoň2022} also de-randomized the previously fastest SSSP algorithm in the CONGEST model. This yields the remaining claims in \Cref{thm:unfiedAppraochTreeEmbeddings}.

\section{Tight $\ell_1$-Oblivious Routing} 
\label{sec:obliviousRouting}

Our description of the construction of the $\ell_1$-oblivious routing follows closely \cite{rozhovn2022deterministic} (we slightly modify their construction to be simpler and yield tight approximation). Our analysis, however, deviates significantly from \cite{rozhovn2022deterministic},
and is perhaps closer to the approach in \cite{zuzic2022universally}. In the construction below, we require an additional SSSP subroutine that returns distance estimates $\widehat{\dist}(s,t)$ for all $t \in V$ such that for every edge $(u,v) \in E$, $|\widehat{\dist}(s,u) -  \widehat{\dist}(s,v)|\leq 2 w(u,v)$. 

\paragraph{Construction of the Oblivious Routing.} For efficiency reasons, we compute the $\ell_1$-oblivious routing $\AA \in \mathbb{R}^{m \times n}$ as a product of two matrices $\BB \in \mathbb{R}^{m \times |\mathcal{P}|}$ and $\CC \in \mathbb{R}^{|\mathcal{P}| \times n}$ where $\mathcal{P}$ is a path collection. The idea is that in $\CC$, we can add an entire path from $\mathcal{P}$ to the column of any vertex $v$ in $\CC$ and thus require only $O(\log n)$ bits instead of $O(|P|\log n)$ bits. For $\BB$, we have to store for each path $P \in \mathcal{P}$ the edge entries explicitly and thus need $O(|P|\log n)$ bits, and thus a total of $O(\sum_{P \in \mathcal{P}} |P| \log n)$ bits. This makes it cheap to use paths $P \in \mathcal{P}$ multiple times across different columns of $\CC$, even if $P$ consists of many edges. 

While it is not possible to efficiently compute $\AA$ explicitly, this yields an efficient way to apply $\AA$, meaning that for any demand $\dd$, we can compute $\AA \dd$ efficiently by first computing $\gg = \CC \dd$ and then $\ff = \BB \gg$. Similarly, we can also compute products $\AA^T \vv$ efficiently. 

\paragraph{Oblivious Routing via (approximate) random-shift decompositions.} Let us next describe the approximate random-shift decomposition that we use in the construction of $\BB$ and $\CC$.  

Let $D = 36 \log(n)$. We compute for every $0 < l < L$, and $0 < d \leq D$ an (approximate) random-shift decomposition $\mathcal{C}_{l,d}$ computed from either \Cref{alg:randomShift} or \Cref{alg:approxRandomShift} invoked with parameter $2^l$. Note that we compute all of these decompositions independently. For each such (sub)partition $\mathcal{C}_{l,d}$, we let $\textsc{Center}_{l,d}$ denote the corresponding center function. For convenience, we define $\textsc{Center}_{0, d}(v) = v$ and $\textsc{Center}_{L, d}(v) = r$ for some arbitrary vertex $r \in V$ and every $v \in V$ and $0 < d \leq D$. In particular, every (approximate) shortest path from a vertex $v \in V$ to one of its centers can be obtained from the concatenation of $O(\log^2 n)$ paths.

\paragraph{Constructing the path collection and $\BB$.} For every (sub)partition $\mathcal{C}_{l,d}$, we let $F_{l,d}$ be the associated forest that was computed by \Cref{alg:randomShift} or \Cref{alg:approxRandomShift} and spans the partition sets. Applying a heavy-light decomposition (see \cite{sleator1981data}), we can decompose $F_{l,d}$ into a collection of edge-disjoint paths $\mathcal{P}_{l,d}$ such that every path in $F_{l,d}$ intersects with $O(\log n)$ such paths only. Finally, for every path $P \in \mathcal{P}_{l,d}$, we add the dyadic sub-segments of $P$ to the collection $\mathcal{P}$ that is for every $0 \leq i \leq \lceil\lg(n)\rceil$ and $j \in (0, \lceil|P|/2^i \rceil)$, we add the segment in $P$ from the $(j \cdot 2^i +1)$-th to the $(\min\{(j+1) \cdot 2^i, |P|\})$-th vertex to $\mathcal{P}$. 

\paragraph{Constructing the routing matrix $\CC$.}
To construct $\CC$, for every vertex $v \in V$, $0 \leq l \leq L$, we define for every cluster $C \in \mathcal{C}_{l,d}$ for $0 < d \leq D$, $p_{l,d, C}(v) = \min\{1, \widehat{\dist}(v, V \setminus C)/2^l\}$ (note that if $v \not\in C$, then $p_{l,d,C}(v) = 0$). We let $w_{l}(v) = \sum_{0 < d \leq D} \sum_{C \in \mathcal{C}_{l,d}} p_{l, d, C}(v)$.

We now construct $\CC$ such that the column for $v \in V$ routes $f_{l, d, d', C, C'}(v) = \frac{p_{l, d, C}(v)}{w_l(v)} \cdot \frac{p_{l+1, d', C'}(v)}{w_{l+1}(v)}$ units of flow along some  $\textsc{Center}_{l,d}(v)$ to $\textsc{Center}_{l+1,d'}(v)$ path, for every $d, d' \in [D]$ and $C \in \mathcal{C}_{l,d}, C' \in \mathcal{C}_{l+1,d'}$. We take the path between the centers to be the (approximate) shortest paths via some vertex $w \in V$ that also needs to route flow between them, i.e. we let \[
m_{l, d,d', C, C'} = \min_{w \in V, f_{l,d,d', C, C'}(w) \neq 0} \dist_{T_{l, d}}(\textsc{Center}_{l,d}(w), w) + \dist_{T_{l+1, d'}}(w, \textsc{Center}_{l+1, d'}(w))
\]
and let $P_{l,d,d', C, C'}$ be the corresponding $\textsc{Center}_{l,d}(v)$ to $\textsc{Center}_{l+1,d'}(v)$ path which can be obtained from the composition of $O(\log^2 n)$ paths in $\mathcal{P}$. For each such segment, we add $f_{l,d,d', C, C'}(v)$ to its entry. This completes the construction of routing matrix $\CC$ and thus of our oblivious routing matrix $\AA$.

\paragraph{Correctness.} We start by establishing that $\AA$ indeed is an oblivious routing, i.e. for every demand $\dd$, flow $\ff = \AA \dd$ routes the demand $\dd$. 

To this end, let us first analyze $\gg = \AA \vecone_v$ for some vertex $v \in V$. We claim that $\gg$ is a distribution over $vr$-paths (recall that $r$ is the arbitrary vertex that is chosen to be the center of all vertices at level $L$). This yields overall correctness since it implies that $\AA \dd(v)$ sends $\dd(u)$ units of flow away from $u$, conserves flow at other vertices, and sends $\dd(u)$ units of flow into $r$. Since $\sum_v \dd(v) = 0$ as $\dd$ is a demand, the positive and negative flows cancel the flow into $r$ (except for $-\dd(r)$) and the result follows.

To see that $\gg = \AA \vecone_v$ is a distribution over $vr$-paths, observe that for level $0 \leq l < L$, we have for every $0 < d \leq D$, $C \in \mathcal{C}_{l,d}$ that the out-flow of $\textsc{Center}_{l,d}(v)$ is 
\begin{align*}
\sum_{0 < d' \leq D, C' \in \mathcal{C}_{l+1, d'}} f_{l, d,d', C, C'}(v) &= \sum_{0 < d' \leq D, C' \in \mathcal{C}_{l+1, d'}} \frac{p_{l, d, C}(v)}{w_l(v)} \cdot \frac{p_{l+1, d', C'}(v)}{w_{l+1}(v)} \\
&= \frac{p_{l, d, C}(v)}{w_l(v)} \cdot \sum_{0 < d' \leq D, C' \in \mathcal{C}_{l+1, d'}} \frac{p_{l+1, d', C'}(v)}{w_{l+1}(v)} = \frac{p_{l, d, C}(v)}{w_l(v)}
\end{align*}
For $0 < l \leq L$, we further have from a similar calculation that the in-flow is $\frac{p_{l, d, C}(v)}{w_l(v)}$ and thus the flow is conserved at each center on level $0 < l < L$. This immediately yields the claim.

\paragraph{Tight Competitive Ratio.} Before we dive into the proof of the competitive ratio, let us first observe that for every vertex $v \in V$, level $0 < l < L$ and $0 < d \leq D$, and either by an argument via \Cref{lem:gap1stvs2nddist} or directly by \Cref{lma:ballInCluster}, we have that with probability at least $1/2$, $v$ is clustered and $B(v, 2^l/8) \subseteq \mathcal{C}_{l,d}(v)$. Thus, $\E[p_{v,l,d}(v)] \geq 1/8$. It follows immediately from a Chernoff bound that $D/16 \leq w_{l}(v) \leq D$ with probability at least $1-n^{-3}$. We condition on this event and the event that if $v$ is clustered by $\mathcal{C}_{l,d}$, then the distance to its center is bounded, i.e. $ \dist(v, \textsc{Center}_{l,d}(v)) \leq 2^l \cdot 10 \log n$, for all choices of $v \in V$, $0 < l < L$ and $0 < d \leq D$. By a union bound, these events occur with probability at least $1-O(n^{-2})$.

Let us now bound the competitive ratio of $\AA$. From matrix norms, we have that the competitive ratio is equivalent to $\max_{(u,v) \in E} \frac{\|\AA (\vecone_u- \vecone_v)\|_1}{\dist(u,v)}$ (see for example \cite{florescu2024optimal}). For the rest of the proof, let us fix $u, v \in V$ so that they achieve the maximum ratio among all vertex pairs. We observe that 
\begin{align}\label{eq:boundOblRouting}
\|\AA (\vecone_u- \vecone_v)\|_1 \leq \sum_{0 \leq l < L} \sum_{\substack{0 < d \leq D, C \in \mathcal{C}_{l,d},\\ 0 < d' \leq D , C' \in \mathcal{C}_{l+1, d'}}} |f_{l, d, d', C, C'}(v) - f_{l, d, d', C, C'}(u)| \cdot w(P_{l,d,d',C,C'})
\end{align}
since each term $|f_{l, d, d', C, C'}(v) - f_{l, d, d', C, C'}(u)|$ yields the amount of flow sent between center $C$ and $C'$ along path $P_{l,d,d',C,C'}$ (from the demand, $v$ uses the path in the forward direction, $u$ in the backwards direction) and we are summing over all combinations of clusters at consecutive levels, and levels. Note that the inequality is due only to further cancellations, i.e. an edge $e \in E$ might occur on two (or more) paths $P_{l,d,d',C,C'}$ and $P_{l',d'',d''',C'',C'''}$ and the flows along these paths goes into the opposite directions, thus yielding a cancellation on $e$.

Let us define for each $0 \leq l < L$, the variable 
\[
Y_l = \sum_{\substack{0 < d \leq D, C \in \mathcal{C}_{l,d},\\ 0 < d' \leq D , C' \in \mathcal{C}_{l+1, d'}}} |f_{l, d, d', C, C'}(v) - f_{l, d, d', C, C'}(u)| \cdot w(P_{l,d,d',C,C'}).
\]
The goal of this section is to argue about the sum $Y = \sum_{0 \leq l < L} Y_l$ and bound the sum by $O(\log n)$ w.h.p. The key to obtaining this bound is to find an upper bound on the expectation of $Y_l$, i.e., $\E[Y_l]$. This allows us to conclude the result by a simple Chernoff bound.

To avoid working with the absolute value when analyzing $Y_l$, we observe
\begin{align}
\begin{split}
&|f_{l, d, d', C, C'}(v) - f_{l, d, d', C, C'}(u)| \\&=
 \max\{0, f_{l, d, d', C, C'}(v) - f_{l, d, d', C, C'}(u)\}+  \max\{0, f_{l, d, d', C, C'}(u) - f_{l, d, d', C, C'}(v)\} 
\end{split}
\end{align}

We next analyze the non-canceled flow sent between a specific pair of level $l$ and $l+1$ centers more carefully. For any $0 < d \leq D, C \in \mathcal{C}_{l,d}, 0 < d' \leq D , C' \in \mathcal{C}_{l+1, d'}$
\begin{align}\label{eq:uncanceledFlow}
\begin{split}
    f_{l, d, d', C, C'}(v) - f_{l, d, d', C, C'}(u)
    &= \frac{p_{l, d, C}(v)}{w_{l}(v)} \cdot \frac{p_{l+1, d', C'}(v)}{w_{l+1}(v)} - \frac{p_{l, d, C}(u)}{w_{l}(u)} \cdot \frac{p_{l+1, d', C'}(u)}{w_{l+1}(u)} \\
    &= O\left(\frac{p_{l, d, C}(v) p_{l+1, d', C'}(v) - p_{l, d, C}(u) p_{l+1, d', C'}(u)}{D^2}\right).
\end{split}
\end{align}
The key observation now is that $\dist(u, V \setminus C) \leq \dist(v, V \setminus C) - \Delta$ for $\Delta = d(u,v)$ by the triangle inequality. This yields
\begin{align*}
\dist(v, V \setminus C) \cdot \dist(v, V \setminus C') - &\dist(u, V \setminus C) \cdot \dist(u, V \setminus C') \\&\leq 
\dist(v, V \setminus C) \cdot \dist(v, V \setminus C') - (\dist(v, V \setminus C)- \Delta) \cdot (\dist(v, V \setminus C')- \Delta)\\
&= \Delta (\dist(v, V \setminus C) + \dist(v, V \setminus C')) - \Delta^2\\
&<  \Delta (\dist(v, V \setminus C) + \dist(v, V \setminus C').
\end{align*}
Note that by definition of $p_{l, d, C}(v)$ (and our assumption on $\widehat{\dist}$), this implies 
\[
p_{l, d, C}(v) p_{l+1, d', C'}(v) - p_{l, d, C}(u) p_{l+1, d', C'}(u) \leq \frac{2\Delta}{2^l} \cdot (p_{l, d, C}(v) + p_{l+1, d', C'}(v))
\]
since capping the values $p_{l, d, C}(v), p_{l+1, d', C'}(v), p_{l, d, C}(u), p_{l+1, d', C'}(u)$ at $1$ can only tighten the inequality. We conclude that 
\begin{align}\label{eq:deterministicUpperBoundOblRouting}
\begin{split}
&\sum_{\substack{0 < d \leq D, C \in \mathcal{C}_{l,d},\\ 0 < d' \leq D , C' \in \mathcal{C}_{l+1, d'}}} |f_{l, d, d', C, C'}(v) - f_{l, d, d', C, C'}(u)| \\
&= O\left( \frac{\Delta}{2^l D^2} \cdot \sum_{\substack{0 < d \leq D, C \in \mathcal{C}_{l,d},\\ 0 < d' \leq D , C' \in \mathcal{C}_{l+1, d'}}} \sum_{z \in \{u,v\}} (p_{l, d, C}(z) + p_{l+1, d', C'}(z))  \right).\\
&= O\left(\frac{\Delta}{2^l D^2} \cdot D \cdot (w_l(v) + w_{l+1}(v) + w_l(u) + w_{l+1}(u))\right)\\
&=O(\Delta/2^l)
\end{split}
\end{align}
Returning to our analysis of $Y_l$, we obtain that for $\Gamma_{l,d}(z) = \dist(z, \textsc{Center}_{l,d}(z))$, we can conclude that $\E[Y_l]$ is bound from above by
\begin{align}
\begin{split}
&\E\left[\sum_{\substack{0 < d \leq D, C \in \mathcal{C}_{l,d},\\ 0 < d' \leq D , C' \in \mathcal{C}_{l+1, d'}}} |f_{l, d, d', C, C'}(u) - f_{l, d, d', C, C'}(v)| \cdot \left(\Gamma_{l,d}(z) + \Gamma_{l+1,d'}(z)\right)\right] \\
&= O(\Delta) \cdot (r_l(v) + r_{l+1}(v) + r_l(u) + r_{l+1}(u))
\end{split}
\end{align}
where we let $r_l(v), r_{l+1}(v), r_l(u), r_{l+1}(u)$ be as in the ball growing process described in \Cref{sec:analysisExact} for level $l$/ $l+1$ and vertex $u$/ $v$, respectively. 

We use next that the inequalities \eqref{eq:deterministicUpperBoundOblRouting} are deterministic and so each variable $Y_l$ has value at most 
\[O(\Delta) \cdot (r_l(v) + X_l(v) + r_{l+1}(v) + X_{l+1}(v) + r_l(u) + X_l(u) + r_{l+1}(u) + X_{l+1}(u))
\]
for $X_l(v), X_{l+1}(v), X_l(u), X_{l+1}(u) \sim \textsc{Exp}(1)$ by the fact that distances-to-centers are stochastically dominated by an $r_l$ term and an exponentially distributed variable $X_l \sim \textsc{Exp}(1)$ as explicitly pointed out both in \Cref{sec:analysisExact} and \Cref{sec:analysisApprox}. 

We have from the Chernoff bound for exponentially-distributed random variables that the sums $X(u) = \sum_{0 \leq l \leq L} X_l(u)$ and  $X(v) = \sum_{0 \leq l \leq L} X_l(v)$ are both upper bounded by $10 \log n$ with probability at least $1-O(n^{-4})$. Conditioning on this event finally yields
\begin{align*}
\sum_{0 \leq l < L} Y_l &= O\left(\sum_{0 \leq l < L} \Delta \cdot (r_l(v) + X_l(v) + r_{l+1}(v) + X_{l+1}(v) + r_l(u) + X_l(u) + r_{l+1}(u) + X_{l+1}(u)) \right)\\
&= O\left(\sum_{0 \leq l \leq L} \Delta \cdot (r_l(v) + X_l(v) + r_l(u) + X_l(u)) \right)\\
&= O\left(\Delta \cdot \left(\left(\sum_{0 \leq l \leq L} r_l(v)\right) + \left(\sum_{0 \leq l \leq L} r_l(u)\right) + \left(\sum_{0 \leq l \leq L} X_l(u)\right) 
 + \left(\sum_{0 \leq l \leq L} X_l(v)\right)\right)\right)\\
 = O(\Delta \log n)
\end{align*}
where in the last inequality, we use \Cref{lem:RSLDD:sum-over-r-is-logarithmic} and our previous conditioning. This concludes the proof.

\paragraph{Implementation.} For the implementation, we refer the reader to \Cref{sec:treeEmbeddings} for the implementation of (approximate) random-shift decomposition. The remaining details can be found in \cite{rozhovn2022deterministic}, where a fast SSSP routine is given that computes the distance function $\widehat{\dist}$ as required. We note that \cite{rozhovn2022deterministic} does not compute the quantities $m_{l,d,d',C, C'}$ since they only route between clusters where $C \subseteq C'$, and thus, the shortest path between centers is already pre-computed. However, each vertex $v \in V$ participates in at most $L \cdot D^2$ such minimization problems; thus, these quantities can be implemented efficiently in their minor-aggregation model of computation, which in turn is efficiently implementable in PRAM and CONGEST models. Analogously, paths $P_{l,d,d', C, C'}$ can be computed efficiently. \Cref{thm:oblviousRouting} follows.

We finally note that our proof can be turned into a verifier of whether $\AA$ is of good quality: we can simply check $\AA(\vecone_u - \vecone_v)$ for all edges $(u,v) \in E$. Using this approach, we can further convert our algorithm from a Monte-Carlo algorithm into a Las-Vegas algorithm.

\section*{Acknowledgment} 

We thank the FOCS reviewers for their careful and insightful feedback. We are especially grateful to Bernhard Haeupler and Arnold Filtser for clarifying the history of the problem and for pointing out the overlap with \cite{haeupler2016faster, czumaj2021exploiting}.

\newpage
\bibliographystyle{alpha}
\bibliography{refs}

\appendix
\section{Related Work}
\label{sec:sota}

In this section, we review related work that was not covered in the introduction.

\paragraph{Low-Stretch Spanning Trees (LSSTs).} The notion of probabilistic tree embeddings can be strengthened to require sampled trees $T$ to be subgraphs of $G$, i.e. $T \subseteq G$. The probabilistic tree embeddings are then commonly referred to as \emph{low-stretch spanning trees}. This enables additional applications, most notably, as pre-conditioners to solve Laplacian systems. In fact, the construction in \cite{buy-at-bulk-awerbuch1997} is a low-stretch spanning tree. In \cite{elkin2005lower}, a construction with polylogarithmic stretch was given, which was subsequently improved to $O(\log n \log\log n)$ \cite{abraham2008nearly, abraham2019using}. 

\paragraph{Fast Algorithms for Tree Embeddings.} While \cite{frttrees2003} already gives an $O(n^2)$ time algorithm for constructing FRT trees, this bound was further improved to $O(m \log^3 n)$ in \cite{mendel2009fast} and then to $O(m \log n)$ by \cite{blelloch2017efficient}. \cite{blelloch2012parallel, blelloch2020parallelism, friedrichs2018parallel} consider the construction of FRT trees in the parallel setting where the best constructions achieve $\tilde{O}(1)$ depth, and either $O(n^2)$ or $O(m^{1+\delta})$ for some constant $\delta > 0$ work. The best parallel algorithm with $\tilde{O}(1)$ depth and $\tilde{O}(m)$ work achieves approximation $O(\log^2 n)$ \cite{blurryballgrowing-becker2019}. In the distributed CONGEST model, \cite{khan2008efficient, ghaffari2014near} yield an algorithm to construct FRT trees in $O(n^{0.5+\delta} + D)$ rounds for any constant $\delta > 0$ where $D$ denotes the unweighted diameter of $G$. From \cite{das2011distributed}, a lower bound of $\Omega(n^{0.5} + D)$ is known. 

LSSTs with $O(\log n\log\log n)$ approximation can be computed in time $O(m \log n \log\log n)$ \cite{abraham2019using}. Constructions achieving polylogarithmic approximation running in work $\tilde{O}(m)$ and depth $\tilde{O}(1)$ are known \cite{rozhovn2022deterministic}. 

\paragraph{Fast Algorithms for Tree Embeddings.} In a recent line of work, various algorithms achieved $\tilde{O}(1)$ competitive ratio via algorithms with $\tilde{O}(m)$ runtime/work \cite{li2020faster, andoni2020parallel, asssp-improved-rozhoň2022, fox2024simple}. Notably, all but the last algorithm can be implemented in parallel with $\tilde{O}(1)$ depth. \cite{asssp-improved-rozhoň2022} can also be implemented in the CONGEST model in $\tilde{O}(\sqrt{n} + D)$ rounds. \cite{asssp-improved-rozhoň2022, fox2024simple} are deterministic constructions.

\paragraph{Ramsey Trees via the FRT framework.} In  \cite{mendel2007ramsey}, Mendel and Naor showed that the FRT tree algorithm can be shown to yield for any $k \geq 1$, and $u \in V$ that the sampled tree $T$ preserves all distance $\dist_G(u,v)$ for $v \in V$ up to a factor $O(k)$ with probability $n^{-1/k}$. This was achieved by generalizing the proof from \cite{frttrees2003}. Such trees $T$ are called Ramsey trees and have since received significant further attention. Ramsey trees have deep applications to distance oracles, oblivious routing problems and spanners (see for example \cite{mendel2007ramsey, chechik2014approximate, chechik2015approximate}). The trade-off between approximation and success probability has since been reduced significantly \cite{blurryballgrowing-becker2019, naor2012scale, abraham2020ramsey}. However, as pointed out in \cite{naor2012scale}, it seems unlikely that the proof from \cite{frttrees2003} can be extended to yield a near-perfect trade-off of $(2k-1)$ stretch for a given probability.

\paragraph{Random-Shift Decompositions.} Random-shift, as proposed by Miller et al.\cite{randomshift2013}, was developed in the context of parallelizing spanner constructions. In \cite{randomshift2013}, for an unweighted graph $G$ and any $k \geq 1$, they give a randomized algorithm to compute a subgraph $H \subseteq G$ with work $O(m)$ and depth $O(k)$ w.h.p. such that $H$ has $O(n^{1+1/k})$ edges and preserves all distances up to an $O(k)$ factor. At a loss of $O(\log k)$, their construction also works for weighted graphs. 

Subsequently, Elkin and Neiman \cite{elkin2018efficient} improved the approximation guarantee to $2k-1$, which is optimal under Erdös' girth conjecture. In \cite{forster2022improved}, Forster, Gr{\"o}sbacher, Martin and de Vos show that by replacing the exponential distribution with a capped geometric distribution further yields that the number of rounds $O(k)$ can be made deterministic.

Recently, random-shift techniques have inspired dynamic algorithms for probabilistic tree embeddings \cite{forster2019dynamic, chechik2020dynamic, forster2021dynamic} and shortest-paths problems \cite{bernstein2020near}.

\section{Chernoff Bound for Sum of Exponentially-Distributed Variables}
\label{sec:chernoff}

Given i.i.d. random variables $X_1, X_2, \ldots, X_n \sim \textsc{Exp}(1)$. Let $X = \sum_i X_i$, and $\E[X] = n$, then
\[
\P[X > (2+\delta) n] = \P[e^{tX} > e^{t (2+\delta) n}] \leq \frac{\E[e^{tX}]}{e^{t(2+\delta) \mu}} = e^{-t(2+\delta) \mu} \cdot \prod_{i} \E[e^{tX_i}]. 
\]
for every $t > 0$, using Markov's inequality and then the independence between the random variables.

It remains to use the moment-generating function $M_{X_i}(t) = \E[e^{tX_i}] = \frac{1}{1 - t}$ which holds for any $t < 1$. For $t = 1/2$, $M_{X_i}(1/2) = 2$ and by independence $\E[e^{tX_i}] = 2^n$. We thus obtain
\[
\P[X > (2+\delta) n] < 2^n \cdot e^{-(2+\delta) n/2} < e^{-\delta n/2}.
 \]

\end{document}